\documentclass[11p,reqno]{amsart}

\usepackage[T1]{fontenc}
\usepackage{graphicx}
\usepackage{amssymb,amsthm,amsmath,mathrsfs,bm,braket,marginnote}
\usepackage{enumerate}
\usepackage{appendix}
\usepackage[colorlinks=true, pdfstartview=FitV, linkcolor=blue, citecolor=blue, urlcolor=blue]{hyperref}
\usepackage{multirow}

\usepackage{pgf}
\usepackage{pgfplots}
\usepackage{tikz}
\usetikzlibrary{arrows,calc}
\usepackage{verbatim}
\usetikzlibrary{decorations.pathreplacing,decorations.pathmorphing}
\usepackage[numbers,sort&compress]{natbib}
\usepackage{dsfont}

\numberwithin{equation}{section}
\newtheorem{theorem}{Theorem}[section]
\newtheorem{lemma}[theorem]{Lemma}
\newtheorem{definition}[theorem]{Definition}
\newtheorem{remark}[theorem]{Remark}

\newtheorem{proposition}[theorem]{Proposition}

 \reversemarginpar
 
 \newcommand{\mt}{\mathbf{t}}
 
 \newcommand{\p}{\partial}

\begin{document}
	
\title[Matrix-valued Cauchy bi-orthogonal polynomials]{Matrix-valued Cauchy bi-orthogonal polynomials and a novel noncommutative integrable lattice}

\subjclass[2020]{39A36,~15A15}
\date{}

\dedicatory{}

\keywords{matrix-valued Cauchy bi-orthogonal polynomials; noncommutative C-Toda lattice; quasideterminants; direct method}

\author{Shi-Hao Li}
\address{Department of Mathematics, Sichuan University, Chengdu, 610064, PR China}
\email{shihao.li@scu.edu.cn}

\author{Ying Shi}
\address{School of Science, Zhejiang University of Science and Technology, Hangzhou, 310023, PR China}
\email{yingshi@zust.edu.cn}

\author{Guo-Fu Yu}
\address{School of Mathematical Sciences, Shanghai Jiaotong University, Shanghai, PR China.}
\email{gfyu@sjtu.edu.cn}

\author{Jun-Xiao Zhao}
\address{School of Mathematics and Statistics, University of Glasgow, University Place, Glasgow G12 8QQ, UK}
\email{Jun-Xiao.Zhao@glasgow.ac.uk}

\begin{abstract}
\noindent Matrix-valued Cauchy bi-orthogonal polynomials were proposed in this paper, together with its quasideterminant expression. It is shown that the coefficients in four-term recurrence relation for matrix-valued Cauchy bi-orthogonal polynomials should satisfy a novel noncommutative integrable system, whose Lax pair is given by fractional differential operators with non-abelian variables.
\end{abstract}

\maketitle

\section{introduction}

Noncommutative generalizations of integrable systems have been widely studied with noncommutative space-time variables, which arise from quantization of the phase space. 
In kinds of literature, there have been a large number of noncommutative generalizations of integrable systems, including KP and its related equations \cite{dimakis05,gilson071,gilson072}, Toda lattice and related equations \cite{etingof97,etingof98,krichever81,li09}, Painlev\'e equations \cite{bertola18,retahk10}, and so on.  

The study of closed expressions for solutions is one of the most important topics in soliton theory and 
Hirota's bilinear method (or so-called direct method) is one of the powerful ways of finding solutions \cite{hirota04}. 
It has been widely shown that solutions of classical integrable systems could be written in terms of determinantal  or Pfaffian $\tau$-functions.
The reason owes to Sato's observation, demonstrating that the Pl\"ucker relation on Grassmannian is an integrable system. 
The Pl\"ucker relations are actually identities for determinants if local coordinates are taken \cite{miwa99}. 
However, as determinants and Pfaffians are usually defined by commutative elements and they could hardly act as $\tau$-functions for noncommutative systems, a new algebraic tool is needed. 
The achievement was made by quasideterminants, which was proposed by Gelfand et al \cite{gelfand91,gelfand92,gelfand05}, and it was shown that quasideterminants could act as solutions for noncommutative integrable systems. 
Besides, quasideterminants take advantages in various fields such as matrix-valued orthogonal polynomials, noncommutative combinatorics, noncommutative symmetric functions, continued fractions with noncommutative coefficients, and so on \cite{di11}. 
Therefore, quasideterminants establish ties for all these subjects.

We mainly consider the connection between a new family of matrix-valued orthogonal polynomials called matrix-valued Cauchy bi-orthogonal polynomials, and related noncommutative integrable systems in this paper. 
The study of matrix-valued orthogonal polynomials could date back to Krein about matrix-valued moment problem \cite{krein49}. 
In recent years, it was found that matrix-valued orthogonal polynomials could work as wave functions for noncommutative Toda and Lotka-Volterra lattices \cite{li21,branquinho20,ismail19}, as well as Painlev\'e equations \cite{cafasso14}. 
Such connections were also applied to noncommutative Hermite-Pad\'e approximation \cite{doliwa22} and Wynn recurrence \cite{doliwa222}. 
By introducing a matrix-valued inner product 
\begin{align*}
\langle f(x),g(x)\rangle=\int_{\mathbb{R}}f(x)W(x)g(x)dx
\end{align*}
where $f(x),\,g(x)\in\mathbb{R}^{p\times p}[x]$ are matrix-valued polynomials and $W(x):\, \mathbb{R}\to\mathbb{R}^{p\times p}$ is a symmetric matrix-valued weight function satisfying $W(x)=W^\top(x)$,  one could define a family of matrix-valued orthogonal polynomials $\{P_n(x)\}_{n\in\mathbb{N}}$ through the orthogonal relation $\langle P_n(x),P^\top_m(x)\rangle=H_n\delta_{n,m}$ for some nonsingular normalization factor $H_n$. 
It was shown that matrix-valued orthogonal polynomials satisfy a three-term recurrence relation
\begin{align*}
xP_n(x)=P_{n+1}(x)+a_nP_n(x)+b_nP_{n-1}(x),\quad a_n,\, b_n\in\mathbb{R}^{p\times p},\quad n\in\mathbb{N}.
\end{align*}
Moreover, if time parameters $\{t_i,i\in\mathbb{N}\}$ are introduced into the weight function $W$ such that $\p_{t_i}W(x;\mt)=x^iW(x;\mt)$, then 
$
\p_{t_1}P_n(x;\mt)=-b_nP_{n-1}(x;\mt),
$
and a noncommutative version of Toda hierarchy could be obtained from compatibility condition $\p_{t_1}(xP_n(x;\mt))=x\p_{t_1}P_n(x;\mt)$ \cite{li21}.

We plan to study a matrix-valued generalization of Cauchy bi-orthogonal polynomials and its related noncommutative integrable systems in this paper. 
In recent years,  with the development of interdisciplinary research on random matrices, orthogonal polynomials and integrable systems, connections between variants of orthogonal polynomials and integrable systems have been made, see e.g. \cite{adler03,adler09,chang18,li19} and references therein. 
Cauchy bi-orthogonal polynomials were proposed in the study of the Degasperis-Procesi peakon problem and later found its applications in random matrix model \cite{lundmark05,bertola09,bertola10}. 
With the Cauchy kernel $(x+y)^{-1}$ and nonnegative weights $\omega_1,\,\omega_2$, an inner product over $\mathbb{R}[x]\times\mathbb{R}[y]\to\mathbb{R}$ could be introduced 
\begin{align}\label{cauchy}
\langle f(x),g(y)\rangle=\int_{\mathbb{R}_+\times \mathbb{R}_+}\frac{f(x)g(y)}{x+y}\omega_1(x)\omega_2(y)dxdy.
\end{align}
From the inner product, one can define Cauchy bi-orthogonal polynomials $\{P_n(x),Q_n(y)\}_{n\in\mathbb{N}}$, which are dual to each other, by orthogonality
\begin{align*}
\langle P_n(x),Q_m(y)\rangle=h_n\delta_{n,m}
\end{align*}
for some nonzero constant $h_n\in\mathbb{R}$. 
In \cite{li19}, it was found that symmetric Cauchy bi-orthogonal polynomials (i.e. when weight functions $w_1$ and $w_2$ are equivalent) are related to an integrable lattice equation, which enjoys the same $\tau$-function with CKP equation\footnote{CKP equation is a 2-dimensional generalization of Kaup-Kuperschmidt equation. It is called CKP equation because it is related to the infinite dimensional Lie algebra of type $C$ according to the classification of classical integrable systems \cite{jimbo83}.}, and thus we call the corresponding lattice equation as the Toda lattice of CKP type or C-Toda lattice for brevity. 
Therefore, it is of interest to consider a matrix-valued version of Cauchy bi-orthogonal polynomials and a noncommutative generalization of C-Toda lattice. 
For this purpose, we introduce a symmetric matrix-valued inner product
\begin{align*}
\langle\cdot,\cdot\rangle: \mathbb{R}^{p\times p}[x]\times \mathbb{R}^{p\times p}[y]\to\mathbb{R}^{p\times p},\quad \langle P(x),Q(y)\rangle=\int_{\mathbb{R}_+\times \mathbb{R}_+}\frac{P(x)W(x)W(y)Q(y)}{x+y}dxdy
\end{align*}
with respect to a matrix-valued weight function $W(x)$.

Therefore, one main purpose of this paper is to study the matrix-valued Cauchy bi-orthogonal polynomials, including their quasi-determinant expressions and corresponding recurrence relations. 
We show that the coefficients of the recurrence relations can also be expressed in terms of quasideterminants. 
The other purpose is to introduce proper time flows, from which a novel noncommutative integrable system is obtained by compatibility condition.
This paper is organized as follows. In Section \ref{sec2}, we give a formal definition for matrix-valued Cauchy bi-orthogonal polynomials and their quasi-determinant expressions. We derive a four-term recurrence relation for the matrix-valued Cauchy bi-orthogonal polynomials and show the coefficients in the recurrence relation could also be written in terms of quasideterminants. In Section \ref{sec3}, we introduce proper time flows and hence derive a novel noncommutative integrable lattice. We provide a Lax pair to this novel noncommutative integrable lattice, and verify it through a direct method by quasideterminant identities. Some concluding remarks are given in Section \ref{sec4}.

\section{Matrix-valued Cauchy bi-orthogonal polynomials}	\label{sec2}
	
The theory of matrix orthogonal polynomials has attracted much attention in different aspects of mathematics such as approximation problems on Riemann sphere \cite{bertola21}, quasi-birth-and-death processes \cite{grunbaum08}, random tiling problems \cite{duits20} and non-commutative integrable systems \cite{alvarez17, li21}. 
Such matrix-valued orthogonal polynomial is a non-commutative generalization of standard orthogonal polynomial and related matrix-valued moment problem was proposed by Krein in late 1940s \cite{krein49}. 
In this setting, a bounded matrix-valued Borel measure $d\mu$ on the real line is considered, which maps from the real line $\mathbb{R}$ to a matrix ring $\mathbb{R}^{p\times p}$. 
According to the Radon-Nikodym theorem, such a measure could be expressed in terms of a weight function $W(x)dx$, and a matrix-valued inner product could be written as 
\begin{align*}
\langle \cdot,\cdot\rangle: \mathbb{R}^{p\times p}[x]\times \mathbb{R}^{p\times p}[x]\rightarrow \mathbb{R}^{p\times p},\quad \langle P(x),Q(x)\rangle=\int_{\mathbb{R}}P(x)W(x)Q(x)dx.
\end{align*}
To connect it with non-abelian integrable lattices, we require that $W(x)$ to be symmetric and there exists a family of matrix-valued orthogonal polynomials $\{P_n(x)\}_{n\in\mathbb{N}}$ satisfying orthogonal relation
\begin{align*}
\langle P_n(x), P_m^\top (x)\rangle=H_n\delta_{n,m},
\end{align*}
where $H_n$ is a nonsingular matrix in $\mathbb{R}^{p\times p}$ and $\top$ represents the transpose of the matrix-valued orthogonal polynomials.

Here we propose a bi-orthogonal generalization of the matrix-valued orthogonal polynomials. 
In general, we could consider an inner product over $\mathbb{R}[x]\times\mathbb{R}[y]$ such that
\begin{align*}
\langle P(x),Q(y)\rangle=\int_{\mathbb{R}\times\mathbb{R}} P(x)W_1(x)K(x,y)W_2(y)Q(y)dxdy,
\end{align*}
where $K(x,y)$ is a matrix-valued kernel or a scalar kernel.
In this paper, we consider a specific scalar Cauchy kernel $(x+y)^{-1}$ and an inner product with respect to this kernel  is defined by
\begin{align}\label{innerproduct}
\langle\cdot,\cdot\rangle:\,\mathbb{R}^{p\times p}[x]\times \mathbb{R}^{p\times p}[y]\to\mathbb{R}^{p\times p},\quad \langle P(x),Q(y)\rangle=\int_{\mathbb{R}_+\times \mathbb{R}_+}\frac{P(x)W_1(x)W_2(y)Q(y)}{x+y}dxdy.
\end{align}
Here we require the integration space is a subset of $\mathbb{R}_+\times \mathbb{R}_+$ to avoid the singularity from the kernel. $W_1(x)$ and $W_2(y)$ are thus some matrix-valued weight functions from $\mathbb{R}_+$ to $\mathbb{R}^{p\times p}$. For example, $W_1$ and $W_2$ could be taken as some matrix-valued Laguerre weights. Clearly, we have the following properties for the inner product
\begin{enumerate}
\item Linearity. For any $L_1,\,L_2,\,R_1,\,R_2\in\mathbb{R}^{p\times p}$, we have
 \begin{align*}
& \langle L_1P_1(x)+L_2P_2(x),Q(y)\rangle=L_1\langle P_1(x),Q(y)\rangle+L_2\langle P_2(x),Q(y)\rangle,\\
 &\langle P(x),Q_1(y)R_1+Q_2(y)R_2\rangle=\langle P(x),Q_1(y)\rangle R_1+\langle P(x),Q_2(y)\rangle R_2.
\end{align*}
\item If $W_1$ and $W_2$ are symmetric, then $\langle P(x),Q(y)\rangle^\top=\langle Q^\top(x),P^\top(y)\rangle$.
\end{enumerate}
Besides, we need two requirements on the inner product. One is the existence and finiteness of the moments 
\begin{align*}
m_{i,j}=\langle x^i\mathbb{I}_p,y^j\mathbb{I}_p\rangle=\int_{\mathbb{R}_+\times\mathbb{R}_+}\frac{x^iy^j}{x+y}W_1(x)W_2(y)dxdy\in\mathbb{R}^{p\times p},\quad i,\,j\in\mathbb{N}.
\end{align*}
We need to require that $W_1(x)$ and $W_2(y)$ should be properly taken such that all $\{m_{i,j}\}_{i,j\in\mathbb{N}}$ exist and are finite. The second is the invertibility of principle minors of moment matrices. For any $n\in\mathbb{N}$, principle minors of moment matrices $\left(
m_{i,j}
\right)_{i,j=0}^{n}$ should be invertible. We call these requirements as moment conditions. 

Now we could give the following definition for matrix-valued Cauchy bi-orthogonal polynomials. 
\begin{definition}
With inner product \eqref{innerproduct} satisfying moment conditions,  a family of monic matrix-valued Cauchy bi-orthogonal polynomials $\{P_n(x), Q_n(x)\}_{n=0}^\infty$ are defined by\footnote{A monic matrix-valued polynomial in the matrix-valued polynomial ring $\mathbb{R}^{p\times p}[x]$ is set to be like $\mathbb{I}_px^{n}+a_{n,n-1}x^{n-1}+\cdots+a_{n,0}$, where $a_{n, i}\in\mathbb{R}^{p\times p}$ for all $i=0, \cdots, n-1$, and $\mathbb{I}_p$ is the unity in the matrix ring.}
\begin{align}\label{or}
\langle P_n(x),Q_m(y)\rangle=H_n\delta_{n,m}
\end{align} for some $H_n\in\mathbb{R}^{p\times p}$.
\end{definition}
In fact, the orthogonal relation \eqref{or} is equivalent to 
\begin{align*}
\langle P_n(x),y^i\mathbb{I}_p\rangle=0,\quad 0\leq i\leq n-1
\end{align*}
according to the linearity of the inner product. 
Moreover, if we assume that
\begin{align*}
P_n(x)=\mathbb{I}_px^n+\xi_{n,n-1}x^{n-1}+\cdots+\xi_{n,0}\in\mathbb{R}^{p\times p}[x],
\end{align*}
then the above orthogonal relation is in fact a combination of linear systems with matrix elements coefficients
\begin{align*}
\xi_{n,0}m_{0,j}+\xi_{n,1}m_{1,j}+\cdots+\xi_{n,n-1}m_{n-1,j}+m_{n,j}=0,\quad j=0,\cdots,n-1.
\end{align*}
Therefore, the existence and uniqueness of bi-orthogonal polynomials are equivalent to the existence and uniqueness for solutions in noncommutative linear systems.
According to the quasideterminant theory proposed by Gelfand et al \cite{gelfand05} (a quick introduction is given in the appendix), if moment conditions are satisfied, then the coefficients of matrix-valued Cauchy bi-orthogonal polynomials are solved to be
\begin{align}\label{xin}
\xi_{n,i}=-\left(
m_{n,0},\cdots,m_{n,n-1}
\right)\left(\begin{array}{ccc}
m_{0,0}&\cdots&m_{0,n-1}\\
\vdots&&\vdots\\
m_{n-1,0}&\cdots&m_{n-1,n-1}
\end{array}
\right)^{-1}e_{i+1}^\top,
\end{align}
where $e_i=(0,\cdots,\mathbb{I}_p,\cdots,0)$ is the block unit vector, whose $i$-th element is the unity  and the others are zeros. 
By substituting those coefficients into the expression of $P_n(x)$, we obtain
\begin{align*}
P_n(x)&=x^n\mathbb{I}_p-\left(
m_{n,0},\cdots,m_{n,n-1}
\right)\left(\begin{array}{ccc}
m_{0,0}&\cdots&m_{0,n-1}\\
\vdots&&\vdots\\
m_{n-1,0}&\cdots&m_{n-1,n-1}
\end{array}
\right)^{-1}\left(\begin{array}{c}
\mathbb{I}_p\\\vdots\\x^{n-1}\mathbb{I}_p
\end{array}
\right)\\
&=\left|\begin{array}{cccc}
m_{0,0}&\cdots&m_{0,n-1}&\mathbb{I}_p\\
\vdots&&\vdots&\vdots\\
m_{n-1,0}&\cdots&m_{n-1,n-1}&x^{n-1}\mathbb{I}_p\\
m_{n,0}&\cdots&m_{n,n-1}&\boxed{x^n\mathbb{I}_p}\end{array}
\right|.
\end{align*}

In a similar manner, we could get an explicit formula for $Q_n(x)$ and
\begin{align*}
Q_n(x)=\left|\begin{array}{cccc}
m_{0,0}&\cdots&m_{0,n-1}&m_{0,n}\\
\vdots&&\vdots&\vdots\\
m_{n-1,0}&\cdots&m_{n-1,n-1}&m_{n-1,n}\\
\mathbb{I}_p&\cdots&x^{n-1}\mathbb{I}_p&\boxed{x^n\mathbb{I}_p}
\end{array}
\right|.
\end{align*}

To find its connection with integrable systems, we need to further require that weight functions $W_1$ and $W_2$ are identical and symmetric, i.e. $W_1=W_2=W_1^\top$. 
With this assumption, we know that moments are symmetric, i.e. $m_{i, j}=m_{j, i}^\top$, and we have the following lemma.
\begin{lemma}
With weight functions set to be identical and symmetric, matrix-valued Cauchy bi-orthogonal polynomials defined in \eqref{or} have the relation
\begin{align*}
Q_n^\top(x)=P_n(x).
\end{align*}
\end{lemma}
\begin{proof}
This is a direct verification that
\begin{align*}
Q_n^\top(x)=x^n\mathbb{I}_p-(m_{0,n}^\top,\cdots,m_{n-1,n}^\top)\left(\begin{array}{ccc}
m_{0,0}^\top&\cdots&m_{n-1,0}^\top\\
\vdots&&\vdots\\
m_{0,n-1}^\top&\cdots&m_{n-1,n-1}^\top
\end{array}
\right)^{-1}\left(\begin{array}{c}
\mathbb{I}_p\\
\vdots\\x^{n-1}\mathbb{I}_p
\end{array}
\right)=P_n(x).
\end{align*}
\end{proof}
Since $Q_n(y)$ could be expressed in terms of $P_n(y)$, there is only one family of symmetric Cauchy bi-orthogonal polynomials. We call this family of polynomials as matrix-valued Cauchy orthogonal polynomials from now on. 

To summarize, we have the following proposition. 
\begin{proposition}
With inner product \eqref{innerproduct} in which $W(x)$ is set to be symmetric, a family of Cauchy orthogonal polynomials $\{P_n(x)\}_{n\in\mathbb{N}}$ is given by the orthogonal relation
\begin{align*}
\langle P_n(x), P_m^\top(y)\rangle=H_n\delta_{n,m},
\end{align*}
where $P_n(x)$ and $H_n$ have the following quasideterminant expressions
\begin{align}\label{qp}
P_n(x)=\left|\begin{array}{cccc}
m_{0,0}&\cdots&m_{0,n-1}&\mathbb{I}_p\\
\vdots&&\vdots&\vdots\\
m_{n-1,0}&\cdots&m_{n-1,n-1}&x^{n-1}\mathbb{I}_p\\
m_{n,0}&\cdots&m_{n,n-1}&\boxed{x^n\mathbb{I}_p}\end{array}
\right|,\quad H_n=\left|\begin{array}{cccc}
m_{0,0}&\cdots&m_{0,n-1}&m_{0,n}\\
\vdots&&\vdots&\vdots\\
m_{n-1,0}&\cdots&m_{n-1,n-1}&m_{n-1,n}\\
m_{n,0}&\cdots&m_{n,n-1}&\boxed{m_{n,n}}\end{array}
\right|
\end{align}
with $H_n=H_n^\top$.
\end{proposition}

\subsection{Recurrence relation for matrix-valued Cauchy orthogonal polynomials}
The most important characterization in the inner product \eqref{innerproduct} is the Cauchy kernel $(x+y)^{-1}$. From this inner product, we have
\begin{align}\label{iteration}
\begin{aligned}
\langle xf(x),g(y)\rangle&=\int_{\mathbb{R}_+\times\mathbb{R}_+}\frac{xf(x)W(x)W(y)g(y)}{x+y}dxdy\\
&=\int_{\mathbb{R}_+}f(x)W(x)dx\int_{\mathbb{R}_+}W(y)g(y)dy-\langle f(x),yg(y)\rangle,
\end{aligned}
\end{align}
which holds for arbitrary $f(x)\in\mathbb{R}^{p\times p}[x]$, $g(y)\in\mathbb{R}^{p\times p}[y]$. 
This formula is useful when inducing corresponding recurrence relations.

\begin{proposition}
The matrix-valued Cauchy orthogonal polynomials $\{P_n(x)\}_{n\in\mathbb{N}}$ satisfy the following recurrence relation
\begin{align}\label{ftr}
x(P_{n+1}(x)+a_nP_n(x))=P_{n+2}(x)+b_nP_{n+1}(x)+c_nP_n(x)+d_nP_{n-1}(x)
\end{align}
for some certain $a_n$, $b_n$, $c_n$, $d_n\in\mathbb{R}^{p\times p}$.
\end{proposition}
\begin{proof}
From the formula \eqref{iteration}, we know that if we take
\begin{align}\label{an}
a_n=-\left(\int_{\mathbb{R}_+}P_{n+1}(x)W(x)dx\right)\left(\int_{\mathbb{R}_+}P_n(x)W(x)dx\right)^{-1},
\end{align}
then 
\begin{align*}
\langle x(P_{n+1}(x)+a_nP_n(x)), Q(y)\rangle=-\langle P_{n+1}(x)+a_nP_n(x),yQ(y)\rangle.
\end{align*}
This is a key ingredient to deriving the recurrence formula. 
Moreover, since $\{P_n(x)\}_{n\in\mathbb{N}}$ form a basis in the Hilbert space, we have the following expansion formula
\begin{align*}
x(P_{n+1}(x)+a_nP_n(x))=P_{n+2}(x)+\sum_{i=0}^{n+1}\alpha_{n,i}P_i(x),\end{align*}
where $ \alpha_{n,i}\in\mathbb{R}^{p\times p}$ are some coefficients. 
Therefore, by taking $a_n$ as in \eqref{an}, we have
\begin{align*}
\langle P_{n+2}(x)+\sum_{i=0}^{n+1}\alpha_{n,i}P_i(x),P_m^\top(y)\rangle&=\langle x(P_{n+1}(x)+a_nP_n(x)), P_m^\top(y)\rangle\\&=-\langle P_{n+1}(x)+a_nP_n(x),yP_m^\top(y)\rangle.
\end{align*}
According to the orthogonal relation, we know that when $m<n-1$, the right-hand side is definitely zero, which means that $\alpha_{n,0}=\cdots=\alpha_{n,n-2}=0$. 
Therefore, the four-term recurrence relation \eqref{ftr} is obtained.

Moreover, we could write down the recurrence coefficients with the help of the inner product. 
By Taking the inner product with $P_{n-1}(y)$, $P_n(y)$ and $P_{n+1}(y)$ respectively, we have
\begin{align*}
d_n&=-\left\langle P_{n+1}(x)+a_nP_n(x),yP_{n-1}^\top(y)\right\rangle H_{n-1}^{-1}=-a_nH_nH_{n-1}^{-1},\\
c_n&=-\left\langle P_{n+1}(x)+a_nP_n(x),yP_n^\top(y)\right\rangle H_n^{-1}=-\left(H_{n+1}+a_n\left\langle P_n(x),yP_n^\top(y)\right\rangle\right)H_n^{-1},\\
b_n&=\left\langle x(P_{n+1}(x)+a_nP_n(x)),P_{n+1}^\top(y)\right\rangle H_{n+1}^{-1}=\left\langle xP_{n+1}(x),P_{n+1}^\top(y)\right\rangle H_{n+1}^{-1}+a_n.
\end{align*}
\end{proof}

It should be emphasized that those coefficients could be written in terms of quasideterminants. 
It is obvious that
\begin{align}\label{vn}
V_n:=\int_{\mathbb{R}_+} P_n(x)W(x)dx=\left|
\begin{array}{cccc}
m_{0,0}&\cdots&m_{0,n-1}&\phi_0\\
\vdots&&\vdots&\vdots\\
m_{n-1,0}&\cdots&m_{n-1,n-1}&\phi_{n-1}\\
m_{n,0}&\cdots&m_{n,n-1}&\boxed{\phi_n}
\end{array}
\right|,\quad \phi_i=\int_{\mathbb{R}_+}x^iW(x)dx\in\mathbb{R}^{p\times p}
\end{align}
according to the quasi-determinant formula \eqref{qp} and the linearity.
To ensure that $V_n$ is well-defined, we assume that $\{\phi_i\}_{i\in\mathbb{N}}$ exist and are finite. Moreover, according to the  exact formula of $a_n$ in \eqref{an}, we should assume that $V_n$ is nonsingular to ensure that it has a proper inverse.

Moreover, if we expand $P_n(x)=x^n\mathbb{I}_p+\xi_{n,n-1}x^{n-1}+\cdots+\xi_{n,0}$, then we have
\begin{align*}
b_n=Z^\top_{n+1}H_{n+1}^{-1}+\xi_{n+1,n}+a_n,\quad c_n=-H_{n+1}H_n^{-1}-a_nZ_nH_n^{-1}-a_nH_n\xi_{n,n-1}^\top H_n^{-1}
\end{align*}
where 
\begin{align*}
Z_{n}=\langle P_n(x),y^{n+1}\mathbb{I}_p\rangle=\left|
\begin{array}{cccc}
m_{0,0}&\cdots&m_{0,n-1}&m_{0,n+1}\\
\vdots&&\vdots&\vdots\\
m_{n-1,0}&\cdots&m_{n-1,n-1}&m_{n-1,n+1}\\
m_{n,0}&\cdots&m_{n,n-1}&\boxed{m_{n,n+1}}
\end{array}
\right|.
\end{align*}
It is of interest to notice that $Z_n$ could be written in terms of $\xi_{n+1,n}$.
\begin{proposition}
It holds that
\begin{align*}
Z_n^\top=-\xi_{n+1,n}H_n.
\end{align*}
\end{proposition}
\begin{proof}
According to \eqref{xin}, it is known that
\begin{align*}
\xi_{n+1,n}=\left|\begin{array}{cccc}
m_{0,0}&\cdots&m_{0,n}&0\\
\vdots&&\vdots&\vdots\\
m_{n,0}&\cdots&m_{n,n}&\mathbb{I}_p\\
m_{n+1,0}&\cdots&m_{n+1,n}&\boxed{0}
\end{array}
\right|,
\end{align*}
then application of the noncommutative Jacobi identity \eqref{ncj} gives the desired formula.
\end{proof}
Therefore, coefficients in the four-term recurrence relation \eqref{ftr} could be expressed in terms of quasideterminants and we have
\begin{align}\label{coefficients}
\begin{aligned}
&a_n=-V_{n+1}V_n^{-1},\quad b_n=-\xi_{n+2,n+1}+\xi_{n+1,n}+a_n,\\
&c_n=-H_{n+1}H_n^{-1}+a_nH_n\left(\xi_{n+1,n}^\top-\xi_{n,n-1}^\top\right)H_n^{-1}, \quad d_n=-a_nH_nH_{n-1}^{-1}.
\end{aligned}
\end{align}

\section{Time evolutions and non-commutative C-Toda lattice}\label{sec3}

This part is devoted to the discussion of how to introduce time flows into the matrix-valued Cauchy bi-orthogonal polynomials. 
In fact, formula \eqref{ftr} indicates the following spectral problem
\begin{align*}
\left(\begin{array}{c}
P_{n+2}(x)\\
P_{n+1}(x)\\
P_n(x)
\end{array}
\right)=\left(\begin{array}{ccc}
x-b_n& a_nx-c_n&-d_n\\
\mathbb{I}_p&0&0\\
0&\mathbb{I}_p&0\end{array}
\right)\left(\begin{array}{c}
P_{n+1}(x)\\
P_{n}(x)\\
P_{n-1}(x)
\end{array}\right),
\end{align*}
and it is expected to find proper time flows such that an integrable system could be derived.

Let's assume that there is a family of odd time variables $\mt=(t_1,t_3,\cdots)$ and put them into the weight functions such that\footnote{Since the commutative case belongs to the CKP hierarchy, we constrain ourselves to the odd flows.}
\begin{align}\label{td}
W(x;\mt)=\exp\left(\sum_{i=1}^\infty t_{2i+1}x^{2i+1}\right)W(x).
\end{align}
In literature, those time variables could be possibly infinite many. 
For convenience, let's consider $t_1$-flow, which results in the first non-trivial example in the hierarchy. 
The role of $t_1$-flow is equivalent to a ladder operator in the orthogonal polynomial theory in some sense since $\p_{t_1}W(x;\mt)=xW(x;\mt)$.

Therefore, the inner product \eqref{innerproduct} becomes a time-deformed inner product
\begin{align*}
\langle f(x),g(y)\rangle=\int_{\mathbb{R}_+\times\mathbb{R}_+}\frac{f(x)W(x;\mt)W(y;\mt)g(y)}{x+y}dxdy,
\end{align*}
and matrix-valued Cauchy bi-orthogonal polynomials are time-dependent, denoted by $\{P_n(x;\mt)\}_{n\in\mathbb{N}}$. Since time variables are added in the weight function, it doesn't make any influence on the spectral part. 
Moreover, the moments $\{m_{i,j}\}_{i,j\in\mathbb{N}}$ satisfy the time evolutions
\begin{align*}
\p_{t_1}m_{i,j}=m_{i+1,j}+m_{i,j+1}=\phi_i\phi_j,
\end{align*}
where $\phi_i$ is given in \eqref{vn}. 
With this assumption, $\xi_{n,n-1}$ and $H_n$ have the following relation.
\begin{proposition}
It holds that
\begin{align}\label{H-xi}
\p_{t_1} H_n=\left(\xi_{n,n-1}-\xi_{n+1,n}\right)H_n+H_n\left(\xi_{n,n-1}^T-\xi_{n+1,n}^T\right).
\end{align}
\end{proposition}
\begin{proof}
With $\p_{t_1}m_{i,j}=m_{i+1,j}+m_{i,j+1}$ and applications of derivative formula \eqref{dqd}, we notice that
\begin{align*}
\p_{t_1}H_n&=\left|\begin{array}{cccc}
m_{0,0}&\cdots&m_{0,n-1}&m_{0,n+1}+m_{1,n}\\
\vdots&&\vdots&\vdots\\
m_{n-1,0}&\cdots&m_{n-1,n-1}&m_{n-1,n+1}+m_{n,n}\\
m_{n,0}&\cdots&m_{n,n-1}&\boxed{m_{n,n+1}}\end{array}
\right|+\text{(transposed part)}\\
&+mM^{-1}(A_1+A_2)M^{-1}m^\top
\end{align*}
where 
\begin{align*}
m=(m_{n,0},\cdots,m_{n,n-1}),\quad
M=\left(m_{i,j}
\right)_{i,j=0}^{n-1},\quad A_1=\left(m_{i+1,j}
\right)_{i,j=0}^{n-1},\quad A_2=\left(m_{i,j+1}
\right)_{i,j=0}^{n-1}.
\end{align*}
By inserting an identity matrix, we have
\begin{align*}
mM^{-1}\left(\sum_{j=1}^n e_je_j^\top\right)A_1M^{-1}m^\top=\sum_{j=1}^n \left|\begin{array}{cc}
M&e_j\\
m&\boxed{0}
\end{array}
\right|\cdot\left|\begin{array}{cccc}
m_{0,0}&\cdots&m_{0,n-1}&m_{0,n}\\
\vdots&&\vdots&\vdots\\
m_{n-1,0}&\cdots&m_{n-1,n-1}&m_{n-1,n}\\
m_{j,0}&\cdots&m_{j,n-1}&\boxed{0}
\end{array}
\right|.
\end{align*}
Noting that if two rows are identical, then the corresponding quasideterminant is zero \cite[Prop. 1.4.6]{gelfand05}, we know that the above formula is equal to
\begin{align*}
-\left|\begin{array}{cccc}
m_{0,0}&\cdots&m_{0,n-1}&m_{1,n}\\
\vdots&&\vdots&\vdots\\
m_{n-1,0}&\cdots&m_{n-1,n-1}&m_{n,n}\\
m_{n,0}&\cdots&m_{n,n-1}&\boxed{0}
\end{array}
\right|+\left|\begin{array}{cc}
M&e_n\\
m&\boxed{0}
\end{array}
\right|\cdot\left|\begin{array}{cc}
M&m^\top\\
m&\boxed{m_{n,n}}
\end{array}
\right|.
\end{align*}
Therefore, 
\begin{align*}
\left|\begin{array}{cccc}
m_{0,0}&\cdots&m_{0,n-1}&m_{0,n+1}+m_{1,n}\\
\vdots&&\vdots&\vdots\\
m_{n-1,0}&\cdots&m_{n-1,n-1}&m_{n-1,n+1}+m_{n,n}\\
m_{n,0}&\cdots&m_{n,n-1}&\boxed{m_{n,n+1}}\end{array}
\right|+mM^{-1}A_1M^{-1}m^\top=Z_n+\xi_{n,n-1}H_n.
\end{align*}
Similarly, the transposed part and $mM^{-1}A_2M^{-1}m^\top$ give the rest part, and the proof is complete.
\end{proof}

Moreover, we have the following derivative formula for the polynomials.
\begin{proposition}
Time-dependent Cauchy orthogonal polynomials $\{P_n(x;\mt)\}_{n\in\mathbb{N}}$ satisfy the evolution equation
\begin{align}\label{te}
\p_{t_1}\left(P_{n+1}(x;\mt)+a_nP_n(x;\mt)\right)=\p_{t_1}\left(\xi_{n+1,n}+a_n\right)P_n(x;\mt),
\end{align}
where $a_n$ is given in \eqref{coefficients}.
\end{proposition}
\begin{proof}
From the orthogonal relation \eqref{or}, we have the following formula
\begin{align}
\langle P_{n+1}(x;\mt)+a_nP_n(x;\mt),P_m^\top(y;\mt)\rangle=H_{n+1}\delta_{n+1,m}+a_nH_n\delta_{n,m}.
\end{align}
By taking the $t_1$-derivative on both sides and noting that
\begin{align}\label{a1}
\langle x(P_{n+1}(x;\mt)+a_nP_n(x;\mt)),P_m^\top(y;\mt)\rangle+\langle P_{n+1}(x;\mt)+a_nP_n(x;\mt),yP_m^\top(y;\mt)\rangle=0.\end{align}
we obtain the formula
\begin{align}\label{d1}
\begin{aligned}
\p_{t_1}H_{n+1}\delta_{n+1,m}+\p_{t_1}(a_nH_n)\delta_{n,m}&=\langle \p_{t_1}(P_{n+1}(x;\mt)+a_nP_n(x;\mt)),P_m^\top(y;\mt)\rangle\\&+\langle P_{n+1}(x;\mt)+a_nP_n(x;\mt),\p_{t_1}P_m^\top(y;\mt)\rangle.
\end{aligned}
\end{align}
Therefore, when $m<n$, it is known that the left-hand side of the above formula is equal to zero. 
Moreover, we have $\text{deg}(\p_{t_1} P_m^\top)<\text{deg}P_m^\top$, and thus 
\begin{align}\label{de2}
\langle \p_{t_1}(P_{n+1}(x;\mt)+a_nP_n(x;\mt)),P_m^\top(y;\mt)\rangle=0,\quad m<n.
\end{align}
On the other hand, since $\{P_n(x;\mt)\}_{n\in\mathbb{N}}$ span a polynomial basis in the Hilbert space, we have
\begin{align*}
\p_{t_1}(P_{n+1}(x;\mt)+a_nP_{n}(x;\mt))=\sum_{k=0}^n \gamma_{n,k}P_k(x;\mt).
\end{align*}
By taking it into \eqref{de2}, we could get $\gamma_{n,0}=\cdots=\gamma_{n,n-1}=0$. 
By comparing the coefficients of $x^n$ on both sides, we get that $\gamma_{n,n}=\p_{t_1}(\xi_{n+1,n}+a_n).$
\end{proof}
\begin{remark}
The reason why we consider $\p_{t_1}(P_{n+1}(x;\mt)+a_nP_n(x;\mt))$ is due to the equation \eqref{a1}.
\end{remark}
Below, we derive an integrable lattice from \eqref{d1}, which is equivalent to the compatibility condition of spectral problem \eqref{ftr} and time evolution \eqref{te}. From \eqref{d1}, it is known that when $m=n$, we have
\begin{align*}
\p_{t_1}(a_nH_n)=\langle\gamma_{n,n} P_n(x;\mt),P_n^\top(y;\mt)\rangle=\p_{t_1}(\xi_{n+1,n}+a_n)H_n,
\end{align*}
and when $m=n+1$, we have
\begin{align*}
\p_{t_1}H_{n+1}=\langle P_{n+1}(x;\mt)+a_nP_n(x;\mt), \p_{t_1}P_{n+1}^\top(y;\mt)\rangle=a_nH_n\p_{t_1}\xi_{n+1,n}^\top.
\end{align*}
To conclude, we have the following proposition.
\begin{theorem}
The non-commutative C-Toda lattice is
\begin{align}\label{nc}
\begin{aligned}
\left\{\begin{array}{l}
a_n\cdot \p_{t_1}H_n=\p_{t_1}\xi_{n+1,n}\cdot H_n,\\
\p_{t_1}H_{n+1}=a_n\cdot H_n\cdot \p_{t_1}\xi_{n+1,n}^\top,\\
\p_{t_1} H_n=\left(\xi_{n,n-1}-\xi_{n+1,n}\right)H_n+H_n\left(\xi_{n,n-1}^\top-\xi_{n+1,n}^\top\right).
\end{array}
\right.
\end{aligned}
\end{align}
with Lax pair
\begin{align*}
&x\left(
P_{n+1}(x;\mt)+a_nP_n(x;\mt)
\right)=P_{n+2}(x;\mt)+\left(-\xi_{n+2,n+1}+\xi_{n+1,n}+a_n\right)P_{n+1}(x;\mt)\\
&\qquad-\left(H_{n+1}H_n^{-1}+a_nH_n\left(-\xi_{n+1,n}^\top+\xi_{n,n-1}^\top\right)H_n^{-1}\right)P_n(x;\mt)-a_nH_nH_{n-1}^{-1}P_{n-1}(x;\mt),\\
&\p_{t_1}(P_{n+1}(x;\mt)+a_nP_n(x;\mt))=\p_{t_1}(\xi_{n+1,n}+a_n)P_n(x;\mt).
\end{align*}
\end{theorem}

It should be noticed that there is also a matrix Lax pair. If we denote
\begin{align*}
\Phi=\left(\begin{array}{c}
P_0(x)\\
P_1(x)\\
\vdots\end{array}
\right),\quad A=\left(\begin{array}{cccc}
a_0&\mathbb{I}_p&&\\
&a_1&\mathbb{I}_p&\\
&&\ddots&\ddots
\end{array}
\right),\quad B=\left(\begin{array}{ccccc}
c_0&b_0&\mathbb{I}_p&&\\
d_1&c_1&b_1&\mathbb{I}_p&\\
&\ddots&\ddots&\ddots&\ddots\end{array}
\right),
\end{align*}
and $C=\text{diag}(e_0,e_1,\cdots)$ where $e_i=\p_{t_1}\xi_{i+1,i}$,
then the four-term recurrence relation and the derivative formula is equivalent to
\begin{align*}
Ax\Phi=B\Phi,\quad A\p_t\Phi=C\Phi,
\end{align*}
and the compatibility condition gives
\begin{align}\label{lax}
\p_{t_1}\mathcal{L}=[\mathcal{N},\mathcal{L}]=\mathcal{N}\mathcal{L}-\mathcal{L}\mathcal{N},
\end{align}
where $\mathcal{L}=A^{-1}B$ and $\mathcal{N}=A^{-1}C$.

\begin{remark}
In the commutative case, if we denote $\tau$-function $\tau_n$ and auxiliary function $\sigma_n$ as
\begin{align*}
\tau_n=\det\left(\begin{array}{ccc}
m_{0,0}&\cdots&m_{0,n-1}\\
\vdots&&\vdots\\
m_{n-1,0}&\cdots&m_{n-1,n-1}
\end{array}
\right),\quad \sigma_n=\det\left(\begin{array}{cccc}
m_{0,0}&\cdots&m_{0,n-1}&\phi_0\\
\vdots&&\vdots&\vdots\\
m_{n-1,0}&\cdots&m_{n-1,n-1}&\phi_{n-1}\\
m_{n,0}&\cdots&m_{n,n-1}&\phi_n
\end{array}\right),
\end{align*}
then we have
\begin{align*}
H_n=\frac{\tau_{n+1}}{\tau_n},\quad a_n=\frac{\tau_{n}\sigma_{n+1}}{\sigma_n\tau_{n+1}},\quad \xi_{n+1,n}=-\frac{1}{2}\frac{\p_{t_1}\tau_{n+1}}{\tau_{n+1}},
\end{align*}
and the non-commutative equation \eqref{nc} is equivalent to\footnote{$D_t$ is the Hirota's bilinear derivative operator defined by $D_t f(t)\cdot g(t)=\frac{\p}{\p s}f(t+s)g(t-s)|_{s=0}$.}
\begin{align}\label{c-toda}
D_t\tau_{n+1}\cdot\tau_n=\sigma_n^2,\quad D_t^2\tau_{n+1}\cdot\tau_{n+1}=4\sigma_{n+1}\sigma_n,
\end{align}
which is the so-called C-Toda lattice whose $\tau$-function acts as that of CKP hierarchy \cite{li19}.
\end{remark}

Below, we verify noncommutative variables in \eqref{nc} admit quasi-determinant forms. 
For this purpose, we observe that the first two equations in \eqref{nc} could be equivalently written as
\begin{align*}
\p_{t_1} H_{n+1}=a_n\cdot\p_{t_1} H_n\cdot a_n^\top,\quad \p_{t_1}\xi_{n+1,n}=a_n\cdot\p_{t_1}H_n\cdot H_n^{-1},
\end{align*}
and we need following lemmas.
\begin{lemma}
It holds that
\begin{align}\label{H-V}
\p_{t_1}H_n=V_nV_n^\top.
\end{align}
\end{lemma}
\begin{proof}
In light of the condition on the derivative of moments as follows
\begin{align*}
\p_{t_1}m_{i,j}=\phi_i\phi_j, \quad i, j=0, 1, \dots, n,
\end{align*}
we can do computations on the derivative of quasideterminant following \eqref{dqd-1} and
\begin{align*}
\p_{t_1}H_n=&\p_{t_1}m_{n,n}
+\left|\begin{array}{cccc}
m_{0,0}&\cdots&m_{0,n-1}&\phi_0\\
\vdots&&\vdots&\vdots\\
m_{n-1,0}&\cdots&m_{n-1,n-1}&\phi_{n-1}\\
m_{n,0}&\cdots&m_{n,n-1}&\boxed{0}\end{array}
\right|~\left|\begin{array}{cccc}
m_{0,0}&\cdots&m_{0,n-1}&m_{0,n}\\
\vdots&&\vdots&\vdots\\
m_{n-1,0}&\cdots&m_{n-1,n-1}&m_{n-1,n}\\
\phi_0&\cdots&\phi_{n-1}&\boxed{0}\end{array}
\right|\nonumber\\
&+\left|\begin{array}{cccc}
m_{0,0}&\cdots&m_{0,n-1}&m_{0,n}\\
\vdots&&\vdots&\vdots\\
m_{n-1,0}&\cdots&m_{n-1,n-1}&m_{n-1,n}\\
\p_{t_1}m_{n,0}&\cdots&\p_{t_1}m_{n,n-1}&\boxed{0}\end{array}
\right|
+\left|\begin{array}{cccc}
m_{0,0}&\cdots&m_{0,n-1}&\p_{t_1}m_{0,n}\\
\vdots&&\vdots&\vdots\\
m_{n-1,0}&\cdots&m_{n-1,n-1}&\p_{t_1}m_{n-1,n}\\
m_{n,0}&\cdots&m_{n,n-1}&\boxed{0}\end{array}
\right|.
\end{align*}
Then together with the following quasi-determinant relations
\begin{align*}
V_n-\phi_n=
\left|
\begin{array}{cccc}
m_{0,0}&\cdots&m_{0,n-1}&\phi_0\\
\vdots&&\vdots&\vdots\\
m_{n-1,0}&\cdots&m_{n-1,n-1}&\phi_{n-1}\\
m_{n,0}&\cdots&m_{n,n-1}&\boxed{0}
\end{array}
\right|
,\quad
V_n^\top-\phi_n=
\left|
\begin{array}{cccc}
m_{0,0}&\cdots&m_{0,n-1}&m_{0,n}\\
\vdots&&\vdots&\vdots\\
m_{n-1,0}&\cdots&m_{n-1,n-1}&m_{n-1,n}\\
\phi_0&\cdots&\phi_{n-1}&\boxed{0}
\end{array}
\right|.
\end{align*}
Moreover, by noting that
\begin{align*}
V_n\phi_n-\p_{t_1}m_{n,n}=
\left|\begin{array}{cccc}
m_{0,0}&\cdots&m_{0,n-1}&\p_{t_1}m_{0,n}\\
\vdots&&\vdots&\vdots\\
m_{n-1,0}&\cdots&m_{n-1,n-1}&\p_{t_1}m_{n-1,n}\\
m_{n,0}&\cdots&m_{n,n-1}&\boxed{0}\end{array}
\right|
\end{align*}
and its transposed version, as well as that $\p_{t_1}m_{n,n}=\phi_n\phi_n$, we know the result is true.
\end{proof}
\begin{lemma}
Regarding with the derivative of $\xi_{n+1,n}$, it holds that
\begin{align}\label{nc-xi}
\p_{t_1}\xi_{n+1,n}=-V_{n+1}V_n^\top H_n^{-1}.
\end{align}
\end{lemma}
\begin{proof}
Following the derivative formula \eqref{dqd-1}, we obtain 
\begin{align*}
\!\!\p_{t_1}\xi_{n+1,n}\!\!=\!\!\left|\begin{array}{cccc}
\!\!m_{0,0}&\cdots&\!\!\!\!m_{0,n}&\!\!\!\!0\\
\vdots&&\vdots&\vdots\\
\!\!m_{n,0}&\cdots&\!\!\!\!m_{n,n}&\!\!\!\!\mathbb{I}_p\\
\!\!\p_{t_1}m_{n+1,0}&\!\!\!\!\cdots&\p_{t_1}m_{n+1,n}&\!\!\!\!\boxed{0}\end{array}\right|
+\left|\begin{array}{cccc}
\!\!m_{0,0}&\cdots&\!\!\!\!m_{0,n}&\!\!\!\!\phi_0\\
\vdots&&\vdots&\vdots\\
\!\!m_{n,0}&\cdots&\!\!\!\!m_{n,n}&\!\!\!\!\phi_n\\
\!\!m_{n+1,0}&\cdots&\!\!\!\!m_{n+1,n}&\!\!\!\!\boxed{0}\end{array}\right|
\left|\begin{array}{cccc}
\!\!m_{0,0}&\cdots&\!\!\!\!m_{0,n}&\!\!\!\!0\\
\vdots&&\vdots&\vdots\\
\!\!m_{n,0}&\cdots&\!\!\!\!m_{n,n}&\!\!\!\!\mathbb{I}_p\\
\!\!\phi_0&\cdots&\!\!\!\!\phi_n&\!\!\!\!\boxed{0}\end{array}\right|.
\end{align*}
To simplify the above equation, we need to rewrite the first term as 
\begin{align*}
\left|\begin{array}{cccc}
\!\!m_{0,0}&\cdots&\!\!\!\!m_{0,n}&\!\!\!\!0\\
\vdots&&\vdots&\vdots\\
\!\!m_{n,0}&\cdots&\!\!\!\!m_{n,n}&\!\!\!\!\mathbb{I}_p\\
\!\!\p_{t_1}m_{n+1,0}&\!\!\!\!\cdots&\p_{t_1}m_{n+1,n}&\!\!\!\!\boxed{0}\end{array}\right|
=\phi_{n+1}\left|\begin{array}{cccc}
\!\!m_{0,0}&\cdots&\!\!\!\!m_{0,n}&\!\!\!\!0\\
\vdots&&\vdots&\vdots\\
\!\!m_{n,0}&\cdots&\!\!\!\!m_{n,n}&\!\!\!\!\mathbb{I}_p\\
\!\!\phi_{0}&\!\!\!\!\cdots&\phi_{n}&\!\!\!\!\boxed{0}\end{array}\right|
\end{align*}
and the first part in second term as $V_{n+1}-\phi_{n+1}$. Therefore we have the formula
\begin{align*}
\!\!\p_{t_1}\xi_{n+1,n}\!\!=V_{n+1}\left|\begin{array}{ccccc}
\!\!m_{0,0}&\cdots&\!\!\!\!m_{0,n-1}&\!\!\!\!m_{0,n}&\!\!\!\!0\\
\vdots&&\vdots&\vdots&\vdots\\
\!\!m_{n-1,0}&\cdots&\!\!\!\!m_{n-1,n-1}&\!\!\!\!m_{n-1,n}&\!\!\!\!0\\
\!\!m_{n,0}&\cdots&\!\!\!\!m_{n,n-1}&\!\!\!\!m_{n,n}&\!\!\!\!\mathbb{I}_p\\
\!\!\phi_0&\cdots&\!\!\!\!\phi_{n-1}&\!\!\!\!\phi_n&\!\!\!\!\boxed{0}\end{array}\right|.
\end{align*}
According to the homological relation of quasideterminant \eqref{homo}, we arrive at 
\begin{align*}
&\left|\begin{array}{ccccc}
\!\!m_{0,0}&\cdots&\!\!\!\!m_{0,n-1}&\!\!\!\!m_{0,n}&\!\!\!\!0\\
\vdots&&\vdots&\vdots&\vdots\\
\!\!m_{n-1,0}&\cdots&\!\!\!\!m_{n-1,n-1}&\!\!\!\!m_{n-1,n}&\!\!\!\!0\\
\!\!m_{n,0}&\cdots&\!\!\!\!m_{n,n-1}&\!\!\!\!m_{n,n}&\!\!\!\!\mathbb{I}_p\\
\!\!\phi_0&\cdots&\!\!\!\!\phi_{n-1}&\!\!\!\!\phi_n&\!\!\!\!\boxed{0}\end{array}\right|\\
\!\!=\!\!&
-\left|\begin{array}{cccc}
\!\!m_{0,0}    &\cdots&\!\!\!\!m_{0,n-1}   &\!\!\!\!m_{0,n}\\
\vdots           &          &\vdots.                &\vdots\\
\!\!m_{n-1,0} &\cdots&\!\!\!\!m_{n-1,n-1}&\!\!\!\!m_{n-1,n}\\
\!\!\phi_0       &\cdots&\!\!\!\!\phi_{n-1}   &\!\!\!\!\boxed{\phi_n}
\end{array}\right|
\left|\begin{array}{cccc}
m_{0,0}&\cdots&m_{0,n-1}&m_{0,n}\\
\vdots&&\vdots&\vdots\\
m_{n-1,0}&\cdots&m_{n-1,n-1}&m_{n-1,n}\\
m_{n,0}&\cdots&m_{n,n-1}&\boxed{m_{n,n}}\end{array}
\right|^{-1}\\
=&-V_n^\top {H_n}^{-1},
\end{align*}
which means the equation \eqref{nc-xi} is satisfied.
\end{proof}
Therefore, by substituting equations \eqref{H-V} and \eqref{nc-xi} into \eqref{nc}, the noncommutative C-Toda lattice could be directly verified. 

At the end of this section, we give a discussion on the $t_3$-flow and higher order flows could be similarly computed. 
By noting that
\begin{align*}
\p_{t_3}\langle f(x;\mt),g(y;\mt)\rangle&=\langle \p_{t_3}f(x;\mt),g(y;\mt)\rangle
+\langle f(x;\mt),\p_{t_3}g(y;\mt)\rangle\\
&\quad+\langle x^3f(x;\mt),g(y;\mt)\rangle+\langle f(x;\mt),y^3g(y;\mt)\rangle,
\end{align*}
where the last two terms could be alternatively written as
\begin{align*}
\int_{\mathbb{R}_+}x^2f(x;\mt)W(x;\mt)dx \int_{\mathbb{R}_+}W(y;\mt)g(y;\mt)dy&-\int_{\mathbb{R}_+}xf(x;\mt)W(x;\mt)dx \int_{\mathbb{R}_+}yW(y;\mt)g(y;\mt)dy\\&+\int_{\mathbb{R}_+} f(x;\mt)W(x;\mt)\int_{\mathbb{R}_+}y^2W(y;\mt)g(y;\mt)dy.
\end{align*}
If we choose $f(x;\mt)$ as $P_{n+1}(x;\mt)+a_nP_n(x;\mt)$, then we could take $g(y;\mt)$ as a proper linear combination for $\{P_m^\top(y;\mt)\}_{m\in\mathbb{N}}$ such that the summand is zero. We have the following proposition.

\begin{proposition}
Denote $U_m^\top=\int_{\mathbb{R}_+}yW(y;\mt)P_m^\top(y;\mt)dy$, and if $\eta_m$ and $\zeta_m$ are solutions of equations
\begin{align*}
&V_m^\top+V_{m-1}^\top \eta_m+V_{m-2}^\top \zeta_m=0,\\
&U_m^\top+U_{m-1}^\top \eta_m+U_{m-2}^\top \zeta_m=0,
\end{align*}
then 
$
g(y;\mt)=P_m^\top(y;\mt)+P_{m-1}^\top(y;\mt)\eta_m+P_{m-2}^\top(y;\mt)\zeta_m
$
satisfies
\begin{align*}
\int_{\mathbb{R}_+}W(y;\mt)g(y;\mt)dy=\int_{\mathbb{R}_+}yW(y;\mt)g(y;\mt)dy=0.
\end{align*}
\end{proposition}
To solve the linear system, $\eta_m$ and $\zeta_m$ have quasideterminant expressions
\begin{align*}
\eta_m=\left|\begin{array}{ccc}
V_{m-2}^\top&V_{m-1}^\top&V_m^\top\\
U_{m-2}^\top&U_{m-1}^\top&U_m^\top\\
0&\mathbb{I}_p&\boxed{0}\end{array}
\right|,\quad\zeta_m=\left|\begin{array}{ccc}
V_{m-2}^\top&V_{m-1}^\top&V_m^\top\\
U_{m-2}^\top&U_{m-1}^\top&U_m^\top\\
\mathbb{I}_p&0&\boxed{0}\end{array}
\right|.
\end{align*}
Moreover, if we take $t_3$-derivative to equation
\begin{align*}
&\left\langle P_{n+1}(x;\mt)+a_nP_n(x;\mt), P^\top_{m}(y;\mt)+P_{m-1}^\top(y;\mt)\eta_m+P_{m-2}^\top(y;\mt)\zeta_m\right\rangle\\&\quad=H_{n+1}\left(\delta_{n+1,m}+\eta_m\delta_{n+1,m-1}+\zeta_m\delta_{n+1,m-2}\right)+a_nH_n\left(\delta_{n,m}+\eta_m\delta_{n,m-1}+\zeta_m\delta_{n,m-2}\right),
\end{align*}
then when $m<n$, 
\begin{align}\label{sub1}
\left\langle \p_{t_3}(P_{n+1}(x;\mt)+a_nP_n(x;\mt)),P^\top_{m}(y;\mt)+P_{m-1}^\top(y;\mt)\eta_m+P_{m-2}^\top(y;\mt)\zeta_m\right\rangle=0,
\end{align}
and when $m=n$, 
\begin{align}\label{sub2}
\left\langle \p_{t_3}(P_{n+1}(x;\mt)+a_nP_n(x;\mt)),P^\top_{n}(y;\mt)+P_{n-1}^\top(y;\mt)\eta_n+P_{n-2}^\top(y;\mt)\zeta_n\right\rangle=\p_{t_3}(a_nH_n).
\end{align}
By writing
\begin{align*}
\p_{t_3}\left(
P_{n+1}(x;\mt)+a_nP_n(x;\mt)
\right)=\sum_{i=0}^n \gamma_{n,i}^{(3)}P_{i}(x;\mt),
\end{align*}
above formulas \eqref{sub1} and \eqref{sub2} give an inhomogeneous linear system with $(n+1)$ variables and $(n+1)$ equations. Therefore, an evolution of the wave function with respect to $t_3$ could be obtained, and the compatibility condition leads to a higher order noncommutative integrable system.

\section{Concluding remarks}\label{sec4}
In this paper, we generalized the Cauchy bi-orthogonal polynomials to the matrix-valued version, and obtain a novel noncommutative integrable system. We provide a Lax pair with noncommutative variables and make use of a quasi-determinant technique to verify the solutions. However, there are some interesting questions to be studied. One is to study the Pfaffian version. It is noted that in \cite{li19}, the solutions of C-Toda lattice \eqref{c-toda} were verified by Pfaffian identities. Therefore, it is expected to obtain a Pfaffian version to verify these integrable systems. Another question is a noncommutative version of fractional differential operator. In \cite{casper21}, a fractional differential operator is an operator that can be represented in the form $A^{-1}B$ for two differential operators $A$ and $B$. Therefore, the Lax equation \eqref{lax} is valid for a fractional differential operator with noncommutative variables. It is interesting to develop a Sato-Krichever theory in noncommutative circumstances.  

\appendix
\section*{Appendix: A quick introduction to quasideterminants}
\setcounter{equation}{0}
\renewcommand\theequation{A.\arabic{equation}}

Quasideterminants, introduced by Gelfand \textsl{et al} in the early 1990s \cite{gelfand91},  are important and effective tools for noncommutative algebra. 
Here we give a quick introduction to quasideterminants, including the basic definitions and some identities that we used in previous sections.
Consider an $n\times n$ matrix $A$ over a noncommutative ring $\mathcal{R}$.
In general, it has $n^2$ quasideterminants.
One may calculate every quasideterminant $|A|_{i,j}$, $1\leq i, j\leq n$, as follows
\begin{align*}
|A|_{i,j}=a_{i,j}-r_i^j\left(A^{i,j}\right)^{-1}c_j^i,
\end{align*}
where $A^{i,j}$, denoting the matrix obtained from $A$ by deleting the $i$th row and $j$th column, is an invertible square matrix over $\mathcal{R}$; $r_i^j$ is the row vector over $\mathcal{R}$ obtained form the $i$th row of $A$ by deleting the $j$th entry; $c_j^i$ is the column vector over $\mathcal{R}$ obtained from the $j$the column of $A$ by deleting the $i$th entry. 
Here the noncommutative ring is supposed to be a matrix ring, with the purpose to be in the frame of matrix-valued measure.
Through reordering rows and columns so that the "expansion point" is bottom right, a more graphic notation for the quasideterminant by boxing the element $a_{i,j}$ can be adopted as follows
\begin{align*}
|A|_{i,j}&=\left|\begin{array}{ccccccc}
a_{1,1}   &\cdots&a_{1,j-1}   &a_{1,j+1}    &\cdots &a_{1,n}    &a_{1,j}  \\
  \vdots   &         &\vdots       &\vdots         &\vdots &               &\vdots \\
a_{i-1,1} &\cdots&a_{i-1,j-1} &a_{i-1,j+1}  &\cdots &a_{i-1,n}  &a_{i-1,j} \\
a_{i+1,1}&\cdots&a_{i+1,j-1}&a_{i+1,j+1} &\cdots &a_{i+1,n} &a_{i+1,j} \\
  \vdots  &          &\vdots       &\vdots         &\vdots &               &\vdots\\
a_{n,1}  &\cdots&a_{n,j-1}    &a_{n,j+1}    &\cdots &a_{n,n}   &a_{n,j} \\
a_{i,1}   &\cdots&a_{i,j-1}     &a_{i,j+1}     &\cdots  &a_{i,n}    &\boxed{a_{i,j}}\\
\end{array}
\right|\\
&=\left|\begin{array}{cc}
A^{i,j} &c_j^i \\
r_i^j   &\boxed{a_{i,j}}\\\end{array}
\right|
=a_{i,j}-r_i^j\left(A^{i,j}\right)^{-1}c_j^i.
\end{align*}

One important application of quasideterminants is to solve the linear system with noncommutative coefficients. 
Let $A=(a_{i,j})$ be an $n\times n$ matrix over a ring $\mathcal{R}$. 
Assume that all the quasideterminants $|A|_{i,j}$ are well-defined and invertible. Then
\begin{align*}
\left\{\begin{array}{c}
a_{1,1}x_1+\cdots+a_{1,n}x_n=\xi_1\\
\vdots\\
a_{n,1}x_1+\cdots+a_{n,n}x_n=\xi_n\end{array}
\right.
\end{align*}
have solution $x_i\in\mathcal{R}$ if and only if
\begin{align*}
x_i=\sum_{j=1}^n|A|_{j,i}^{-1}\xi_j.
\end{align*}

Below we show some important identities of quasideterminants \cite{gelfand91, gelfand92, gelfand05}, which are useful in the verifications of solutions for noncommutative integrable systems. 
One is the non-commutative Jacobi identity, which can be written as 
\begin{align}\label{ncj}
\left|\begin{array}{ccc}
A &B &C \\
D &f  &g  \\
E &h &\boxed{i}\\
\end{array}
\right|
=\left|\begin{array}{cc} A &C \\ E &\boxed{i}\\ \end{array} \right|
-\left|\begin{array}{cc} A &B \\ E &\boxed{h}\\ \end{array} \right|
\left|\begin{array}{cc} A &B \\ D &\boxed{f}\\ \end{array} \right|^{-1}
\left|\begin{array}{cc} A &C \\ D &\boxed{g}\\ \end{array} \right|.
\end{align}
where $A$ should be invertible. 
By utilizing the noncommutative Jacobi identity, one could obtain the so-called homological relation
\begin{align}
\begin{aligned}
\label{homo}
&\left|\begin{array}{ccc}
A&B&C\\
D&f&g\\
E&\boxed{h}&i
\end{array}
\right|=\left|\begin{array}{ccc}
A&B&C\\
D&f&g\\
E&h&\boxed{i}
\end{array}
\right|\left|\begin{array}{ccc}
A&B&C\\
D&f&g\\
0&\boxed{0}&1
\end{array}
\right|,\\
&\left|\begin{array}{ccc}
A&B&C\\
D&f&\boxed{g}\\
E&{h}&i
\end{array}
\right|=\left|\begin{array}{ccc}
A&B&0\\
D&f&\boxed{0}\\
E&h&1
\end{array}
\right|\left|\begin{array}{ccc}
A&B&C\\
D&f&g\\
E&h&\boxed{i}
\end{array}
\right|.
\end{aligned}
\end{align}

Besides, it is necessary to consider the derivative formulae for quasideterminants.
Considering the expansion of quasideterminant
\begin{align*}
\left|\begin{array}{cc}
A&B\\
C&\boxed{d}\end{array}
\right|=d-CA^{-1}B,
\end{align*}
and assume that $A,\,B,C$ and $d$ are $t$-dependent,
then we have the general derivative formula
\begin{align}\label{dqd}
\left|\begin{array}{cc}
A &B \\
C &\boxed{d} 
\end{array}
\right|'
=d'-C'A^{-1}B-CA^{-1}B'+CA^{-1}A'A^{-1}B,
\end{align}
where $'$ means the derivative with respect to $t$. 
Especially, if the quasideterminant is well structured,  then derivative formulae will be modified. 
One is the noncommutative version of Grammian determinant.
If $A$ is a Grammian-like matrix satisfying $A'=\sum_{i=1}^{k}E_iF_i$ where $E_i$ (resp. $F_i$) is a  column (resp. row) vector of appropriate length. Then we can factorise the RHS of \eqref{dqd} to obtain 
\begin{align}\label{dqd-1}
\left|\begin{array}{cc}
A &B \\
C &\boxed{d} 
\end{array}
\right|'=d'+\left|\begin{array}{cc}
A &B \\
C' &\boxed{0} 
\end{array}
\right|+\left|\begin{array}{cc}
A &B' \\
C &\boxed{0} 
\end{array}
\right|+
\sum_{i=1}^{k}\left|\begin{array}{cc}
A &E_i \\
C &\boxed{0} 
\end{array}
\right|
\left|\begin{array}{cc}
A &B\\
F_i &\boxed{0} 
\end{array}
\right|.
\end{align}
Another is the noncommutative version of the Wronskian determinant. 
By inserting the identity matrix expressed as $\sum_{j=1}^{n}e_j^\top e_j$, where $e_j=(0,\cdots,\mathbb{I}_p,\cdots,0)$ is the block unit vector, then the above derivative for quasideterminants can be written as
\begin{align}\label{dqd-2}
\left|\begin{array}{cc}
A &B \\
C &\boxed{d} 
\end{array}
\right|'
=\left|\begin{array}{cc}
A &B \\
C' &\boxed{d'} 
\end{array}
\right|
+\sum_{j=1}^{n}\left|\begin{array}{cc}
A &e_j^\top \\
C &\boxed{0} 
\end{array}
\right|
\left|\begin{array}{cc}
A &B\\
\left(A^j\right)' &\boxed{\left(B^j\right)'} 
\end{array}
\right|
\end{align}
and
\begin{align}\label{dqd-3}
\left|\begin{array}{cc}
A &B \\
C &\boxed{d} 
\end{array}
\right|'=\left|\begin{array}{cc}
A &B' \\
C &\boxed{d'} 
\end{array}
\right|
+\sum_{j=1}^{n}\left|\begin{array}{cc}
A &\left(A_j\right)' \\
C &\boxed{\left(C_j\right)'} 
\end{array}
\right|
\left|\begin{array}{cc}
A &B\\
e_j &\boxed{0} 
\end{array}
\right|
\end{align}
where $A^j=e_j A$ and $A_j=A e_j^\top$ denote the $j$th row and column of matrix $A$ respectively.

\section*{Acknowledgement}
This work is partially funded by grants (NSFC12101432, NSFC12175155, NSFC11971322).

\end{document}